\documentclass[12pt]{article}

\usepackage{color}
\usepackage{amsmath}
\usepackage{amsthm}

\newtheorem{lemma}{Lemma}

\newcommand{\algdef}[1]{\textsc{#1}\index{algorytm!\textsc{#1}}\label{algdef:#1}}
\newcounter{lnoc}
\newenvironment{algorithm}[1]{%
\hrule height 0.8pt \vspace{0.6ex} \small#1\vspace{0.6ex}\hrule height 0.5pt \vspace{-2.5ex}
\setcounter{lnoc}{0}
\small
\begin{tabbing}
00000\=XXI\=XXI\=XXI\=XXI\=XXI\=XXI\=\kill
}{%
\end{tabbing}
\vspace{-2.5ex}\hrule height 0.8pt\vspace{1ex}}
\newcommand{\lno}[1][0]{{\footnotesize\sffamily 
\ifnum#1=0
\stepcounter{lnoc} 
\ifnum\thelnoc<10
\phantom0%
\fi
\thelnoc
\else
\thelnoc.#1
\fi
}\>}

\newcommand{\pcfor}{{\bfseries for~}}

\newcommand{\pcto}{{\bfseries to~}}
\newcommand{\pcdownto}{{\bfseries downto~}}
\newcommand{\pcdo}{{\bfseries do~}}
\newcommand{\pcif}{{\bfseries if~}}
\newcommand{\pcthen}{{\bfseries then~}}
\newcommand{\pcelse}{{\bfseries else~}}

\newcommand{\pcand}{{\bfseries and~}}

\newcommand{\pcreturn}{{\bfseries return~}}

\newcommand{\pccomment}[1]{\phantom{000}\{\textit{#1}\}}

\title{Quadratic-time Algorithm\\for the String Constrained LCS Problem}
\author{Sebastian Deorowicz}
\begin{document}
\maketitle

%*********************************************************************************************
\begin{abstract}
The problem of finding a longest common subsequence of two main sequences with some constraint that must be a substring of the result (STR-IC-LCS) was formulated recently.
It is a variant of the constrained longest common subsequence problem.
As the known algorithms for the STR-IC-LCS problem are cubic-time, the presented quadratic-time algorithm is significantly faster.
\end{abstract}

%*********************************************************************************************

%*********************************************************************************************
\section{Introduction}
%The knowledge of the similarity of two or more sequences is crucial in various applications.
%The basic problem, what similarity measure should be used, is dependent on what the sequences contain.
%Some of the possible measures are: edit distance (assuming some set of basic edit operations like: insertion/deletion of symbol, substitution of symbols), longest common substring (a \emph{substring} can be obtained from a sequence by removing any number of symbols (including 0) from beginning and end) or subsequence (a \emph{subsequence} can be obtained from a sequence by removing any number of symbols (also including 0) from any positions)~\cite{Gus1997}.

One of the most popular ways of measuring sequence similarity is computation of their longest common subsequence (LCS)~\cite{Gus1997}, in which we are interested in a subsequence that is common to all sequences and has the maximal possible length.
It is well known that for two sequences of length $n$ and $m$ an LCS can be found in $O(nm)$ time, which is a lower bound of time complexity in the comparison-based computing model for this problem~\cite{AHU1976}.
In the more practical, RAM model of computations, the asymptotically fastest algorithm is the one by Masek and Paterson which runs in $O(nm/\log n)$ time for bounded and $O(mn \log\log n / \log n)$ for unbounded alphabet~\cite{MP1980}.
%The known lower bound in this model~\cite{Hir1978} is, however, $\Omega(n \log n)$ for an unbounded alphabet and $\Omega(n)$ for bounded alphabet, so there is a huge gap between the fastest algorithm and the theoretical lower bound.

One family of LCS-related problems considers one or more {\em constraining} sequences, such that (in some variants) must be included, or (in other problem variants) are forbidden as part of the resulting sequence~\cite{CC2011,Tsa2003}.
The motivation for these generalizations came from bioinformatics in which some prior knowledge is often available and one can specify some requirements on the result~\cite{Tsa2003,Deo2010a}.

In this work, we consider the problem called STR-IC-LCS, introduced in~\cite{CC2011}, in which a constraining sequence of length~$r$ must be included as a substring of a common subsequence of two main sequences and the length of the result must be maximal.
In~\cite{CC2011} an $O(nmr)$-time algorithm was given for it.
Farhana et al.~\cite{FFMR2010} proposed finite-automata-based algorithms for the STR-IC-LCS, CLCS, and two other problems defined by Chen and Chao~\cite{CC2011}.
The authors claim that the algorithms work in $O(r(n+m)+ (n+m)\log (n+m))$ time in the worst case.
It seems to be a breakthrough as it means also that the LCS problem could be solved in $O(n\log n)$ time.
Unfortunately, the time complexity analysis are based on the claim from~\cite{BY1991} that a directed acyclic subsequence graph (DASG) for two sequences of lengths $n$ and $m$ contains $O(n + m)$ states and can be built in $O((n+m)\log(n+m))$ time.
As was shown by Crochemore et al.~\cite{CMT2003} this result was wrong and such a DASG contains $\Omega(nm)$ states in the worst case, so its construction time cannot be lower.
Thus, the algorithms by Farhana et al.~\cite{FFMR2010} work in $\Omega(nmr)$ time for the variants of the CLCS problem and $\Omega(nm)$ for the LCS problem.
Moreover, these complexities are under assumption that the alphabet size is constant, otherwise they should be multiplied by its size.

%The problem was formulated by Tsai~\cite{Tsa2003}, who also showed how to solve it in $O(m^2 n^2 r)$ time, for the two main sequences case.
%Much faster algorithms were proposed soon, e.g.,~\cite{CSFHK2004,AE2005}.
%They worst-case time complexity are $O(mnr)$, which is currently the best result if the time complexity is expressed only in terms of input sequence lengths.
%If we allow to express the complexity in terms of some properties of the input and output sequences better results can be obtained, e.g., $O((m\ell + d)r +n)$ time~\cite{Deo2007} and $O(\sqrt{mnr D} + \lceil{m/w}\rceil nr)$ time~\cite{Deo2010a}, both in the worst case, where  $D$ is the number of so-called strong matches (triples $(i,j,k)$ such that $a_i = b_j = p_k$), $d$ is the number of matches between the two main sequences, $w$ is the processor word size, and  $\ell$ is the result length, 
%The experiments discussed in~\cite{DO2010,Deo2010a} show that these two algorithms are the fastest in practice.

In this paper, we propose the first quadratic-time algorithm for the STR-IC-LCS problem and show also further possible improvements of the time complexity.
We also present how this algorithm can be extended to many main input sequences.

The paper is organized as follows.
In Section~\ref{sec:defs}, some definitions are given and the problem is formally stated.
Section~\ref{sec:alg} describes our algorithm.
Extension to the case of many main sequences and some improvements of the algorithm are given in Section~\ref{sec:impr}.
The last section concludes.

%*********************************************************************************************
\section{Definitions}
\label{sec:defs}
Let us have two main sequences $A = a_1 a_2 \ldots a_n$ and $B = b_1 b_2 \ldots b_m$ and one constraining sequence $P = p_1 p_2 \ldots p_r$.
W.l.o.g.\ we can assume that $r \le m \le n$.
Each sequence is composed of symbols from \emph{alphabet} $\Sigma$ of size $\sigma$.
The \emph{length} (or \emph{size}) of any sequence~$X$ is the number of elements it is composed of and is denoted as $|X|$.
A sequence~$X^\star$ is a subsequence of $X$ if it can be obtained from $X$ by removing zero or more symbols.
The LCS problem for $A$ and $B$ is to find a subsequence~$C$ of both $A$ and~$B$ of the maximal possible length.
The LCS length for $A$ and~$B$ is denoted by~$\mathit{LLCS}(A, B)$.
A sequence $\beta$ is a \emph{substring} of $X$ if $X = \alpha \beta \gamma$ for some, possibly empty, sequences $\alpha$, $\beta$, $\gamma$.
An \emph{appearance} of sequence $X = x_1 x_2\ldots x_{|X|}$ in sequence $Y = y_1 y_2 \ldots y_{|Y|}$ starting at position $j$ is a sequence of indexes $i_1, i_2, \ldots, i_{|X|}$ such that $i_1 = j$, and $X = y_{i_1} \ldots y_{i_{|X|}}$.
A \emph{compact} appearance of a sequence $X$ in $Y$ starting at position $j$ is the appearance of the smallest last index, $i_{|X|}$.
A \emph{match} for sequences $A$ and $B$ is a pair $(i, j)$ such that $a_i = b_j$.
The total number of matches for $A$ and $B$ is denoted by~$d$.
It is obvious that $d \le mn$.

The STR-IC-LCS problem for the main sequences $A$, $B$, and the constraining sequence~$P$ is to find a subsequence $C$ of both $A$ and~$B$ of the maximal possible length containing $P$ as its substring.
(In the CLCS problem, $C$ must be a subsequence of $P$.)

%*********************************************************************************************
\section{The algorithm}
\label{sec:alg}
The algorithm we propose is based on dynamic programming with some preprocessing.
To show its correctness it is necessary to prove some lemma.

Let $C = c_1 c_2 \ldots c_\ell$ be a longest common subsequence with substring constraint for $A$, $B$, and $P$.
Let also $I = (i_1, j_1), (i_2, j_2), \ldots, (i_\ell, j_\ell)$ be a sequence of indexes of $C$ symbols in $A$ and~$B$, i.e., $C = a_{i_1} a_{i_2} \ldots a_{i_\ell}$ and $C = b_{j_1} b_{j_2} \ldots b_{j_\ell}$.
From the problem statement, there must exists such $q \in [1, \ell-r+1]$ that $P = a_{i_q} a_{i_{q+1}} \ldots a_{i_{q+r-1}}$ and  $P = b_{j_q} b_{j_{q+1}} \ldots b_{j_{q+r-1}}$.

\begin{lemma}\label{lem:i-prime}
Let $i^\prime_q = i_q$ and for all $t \in [1, r-1]$, $i^\prime_{q+t}$ be the smallest possible, but larger than $i^\prime_{q+t-1}$, index in $A$ such that $a_{i_{q+t}} = a_{i^\prime_{q+t}}$
The sequence of indexes $I^\prime = (i_1, j_1), \ldots, (i_{q-1}, j_{q-1}),\allowbreak{} (i^\prime_q, j_q), (i^\prime_{q+1}, j_{q+1}), \ldots, (i^\prime_{q+r-1}, j_{q+r-1}), (i_{q+r}, j_{q+r}), \ldots, (i_\ell, j_\ell)$ defines a longest common subsequence of~$A$ and~$B$ with string constraint~$P$ equal~$C$.
\end{lemma}

\begin{proof}
From the definition of indexes $i^\prime_{q+t}$ it is obvious that they form an increasing sequence, since $i^\prime_q = i_q$, and $i^\prime_{q+r-1} \le i_{q+r-1}$.
The sequence $i^\prime_{q}, \ldots, i^\prime_{q+r-1}$ is of course a compact appearance of $P$ in~$A$ starting at $i_q$.
Therefore, both components of $I^\prime$ pairs form increasing sequences and for any $(i^\prime_u, j^\prime_u)$, $a_{i^\prime_u} = b_{j^\prime_u}$, so sequence $I^\prime$ defines an STR-IC-LCS $C^\prime$ equal~$C$.
\end{proof}

A similar lemma can be formulated for $j$-th component of sequence~$I$.
Thus, it is easy to conclude that when looking for an STR-IC-LCS, instead of checking any common subsequences of $A$ and~$B$ it suffices to check only such common subsequences that contain compact appearances of $P$ both in $A$ and~$B$.
(This is a direct consequence of the fact that $\mathit{LLCS}(X, Y) \le \mathit{LLCS}(X, \alpha Y)$ for any sequence~$\alpha$.)

The number of different compact appearances of $P$ in $A$ and $B$ will be denoted by $d^\text{A}$ and~$d^\text{B}$, respectively.
It is easy to notice that $d^\text{A} d^\text{B} \le d$, since a pair $(i, j)$ defines a compact appearance of $P$ in $A$ starting at $i$-th position and compact appearance of $P$ in $B$ starting at $j$-th position only for some matches.

The algorithm computing an STR-IC-LCS (Fig.~\ref{psc:alg}) consists of three main stages.
In the first stage, both main sequences are preprocessed to determine for each occurrence of the first symbol of $P$, the index of the last symbol of a~compact appearance of~$P$.
In the second stage, two DP matrices are computed: the forward one and the reverse one.
The recurrence is exactly as for the LCS computation.

\begin{figure}[t]
\begin{algorithm}{\algdef{STR-IC-LCS}($A$, $B$, $P$)}
%\pcinput{$A$, $B$ --- main sequences}\\
%\pcinpute{$P$ --- constraining sequence}\\
%\pcoutput{Length of a STR-IC-LCS and this sequence}\\
%\pcendparam
\pccomment{Preprocessing}\\
\lno	\pcfor $i \gets 1$ \pcto $n$ \pcdo\\
\lno\>	\pcif $a_i = p_1$ \pcthen
				$M^\text{A}[i] \gets \text{ smallest $q$ such that $p_1\ldots p_r$ is a subsequence of $a_i \ldots a_q$}$\\
\lno	\pcfor $j \gets 1$ \pcto $m$ \pcdo\\
\lno\>	\pcif $b_i = p_1$ \pcthen
				$M^\text{B}[j] \gets \text{ smallest $q$ such that $p_1\ldots p_r$ is a subsequence of $b_j \ldots b_q$}$\\
\pccomment{Computation of forward and reverse DP matrices}\\
\lno	\pcfor $i \gets 0$ \pcto $n+1$ \pcdo
			$F[i, 0] \gets 0$; $R[i, m+1] \gets 0$\\
\lno	\pcfor $j \gets 0$ \pcto $m+1$ \pcdo
			$F[0, j] \gets 0$; $R[n+1, j] \gets 0$\\
\lno	\pcfor $i \gets 1$ \pcto $n$ \pcdo\\
\lno\>	\pcfor $j \gets 1$ \pcto $m$ \pcdo\\
\lno\>\>		\pcif $a_i = b_j$ \pcthen
					$F[i,j] = F[i-1,j-1] + 1$\\
\lno\>\>		\pcelse
					$F[i,j] = \max(F[i-1,j], F[i,j-1])$\\
\lno	\pcfor $i \gets n$ \pcdownto $1$ \pcdo\\
\lno\>	\pcfor $j \gets m$ \pcdownto $1$ \pcdo\\
\lno\>\>		\pcif $a_i = b_j$ \pcthen
					$R[i,j] = R[i+1,j+1] + 1$\\
\lno\>\>		\pcelse
					$F[i,j] = \max(R[i+1,j], R[i,j+1])$\\
\pccomment{Determination of the result}\\
\lno	$\ell \gets 0$; $i^\star \gets 0$; $j^\star \gets 0$\\
\lno	\pcfor $i \gets 1$ \pcto $n$ \pcdo\\
\lno\>	\pcfor $j \gets 1$ \pcto $m$ \pcdo\\
\lno\>\>		\pcif $a_i = b_j$ \pcand
					$F[i,j] + R[M^\text{A}[i], M^\text{B}[j]] + r - 2 > \ell$ \pcthen\\
\lno\>\>\>		$\ell \gets F[i,j] + R[M^\text{A}[i], M^\text{B}[j]] + r - 2$\\
\lno\>\>\>		$i^\star \gets i$; $j^\star \gets j$\\
\lno	Backtrack from $(i^\star, j^\star)$ according to $F$ and obtain $S^1$\\
\lno	Backtrack from $(M^\text{A}[i^\star], M^\text{b}[j^\star])$ according to $R$ and obtain $S^2$\\
\lno	\pcreturn $\ell$ and $S^1 p_2 p_3\ldots p_{r-1} S^2$
\end{algorithm}
\caption{A pseudocode of the STR-IC-LCS computing algorithm for two main sequences and one constraining sequence}
\label{psc:alg}
\end{figure}

In the last stage, the result is determined.
To this end for each match $(i, j)$ for $A$ and $B$ the ends $(i^\prime, j^\prime)$ of compact appearances of $P$ in $A$ starting at $i$-th position and in $B$ starting at $j$-th position are read.
The length of an STR-IC-LCS containing these appearances of $P$ is determined as a sum of the LCS length of prefixes of $A$ and $B$ ending at $i$-th and $j$-th positions, respectively, the LCS length of suffixes of $A$ and $B$ starting at $i^\prime$-th and $j^\prime$-th positions, respectively, and the constraint length.
Since, the first and last constraint symbol was summed twice, the final result is decreased by~2.
According to the $F$ and $R$ matrices, backtracking can be used to obtain the subsequence, not only its length.

\begin{lemma}
The STR-IC-LCS algorithm (Fig.~\ref{psc:alg}) correctly computes an STR-IC-LCS.
\end{lemma}

\begin{proof}
The algorithm considers all pairs of compact appearances of $P$ in $A$ and $B$.
Each such a pair divides the problem into two independent LCS length-computing subproblems.
According to the precomputed $F$ and $R$ matrices it is easy to solve these subproblems (lines 18--20) in constant time.
The length of an STR-IC-LCS must be a sum of the found lengths of LCSs and the constraint length subtracted by~2.
\end{proof}

\begin{lemma}
The worst-case time complexity of the proposed algorithm is $O(mn)$.
\end{lemma}

\begin{proof}
The preprocessing stage can be done in $O((n+m)r)$ worst-case time.
The main stage consists of computation of two DP matrices which needs $O(mn)$ time.
In the final stage, the DP matrix is traversed and for each match a constant number of operations is performed, so these stages consumes $O(mn)$ time.
Summing these up gives $O(mn)$ time.
\end{proof}

\begin{lemma}
The space consumption of the algorithm is $O(mn)$.
\end{lemma}

%*********************************************************************************************
\section{Improvements and extensions}
\label{sec:impr}
%\subsection{Memory reduction}
If one is interested only in the STR-IC-LCS length, it is easy to notice that $F$ and $R$ matrices can be computed row-by-row which means that only $O(m)$ words are necessary for them.
The values of $F$ and~$R$ for matches of symbols equal $p_1$ in~$F$ and $p_r$ in~$R$ must, however, be stored explicitly, so the space for them is $O(d^\star)$ (where $d^\star = O(mn)$ is the number of such matches).
This gives the total space $O(n+d^\star)$.
If also the subsequence is requested, the cells for all matches must be stored to allow backtracking, so the space is $O(n+d)$ (in the worst case $d = O(mn)$).

%\subsection{Hunt--Szymanski speedup}
As the only cells that are necessary to be stored explicitly are those for matches, the Hunt--Szymanski method~\cite{Gus1997} can be used to speed up the computation of $F$ and $R$ matrices if the number of matches is small.
Therefore, the second stage can be completed in $O(d\log\log m + n)$ time if $\sigma = O(n)$ and $O(d\log\log n + n\log n)$ otherwise.
The time complexity of the final stage is~$O(d)$.
Adding the time for the preprocessing we obtain the worst case time complexities:
$O(d\log\log m + nr)$ for $\sigma = O(n)$ and $O(d\log\log n + n(r+\log n))$ otherwise.

%\subsection{Generalization to more main sequences}
The generalization of the LCS problem for many sequences is direct, but the time complexity of the exact algorithm computing the multidimensional DP matrix is $O(2^z n^z)$, where $z$ is the number of sequences of length~$O(n)$ each~\cite{Gus1997}.
It is easy to notice that according to Lemma~\ref{lem:i-prime} the STR-IC-LCS problem generalizes in the same way and the worst-case time complexity is also $O(2^z n^z)$.

%*********************************************************************************************
\section{Conclusions}
\label{sec:concs}
We investigated the STR-IC-LCS problem introduced recently.
The fastest algorithms solving this problem known to date needed cubic time in case of two main and one constraining sequences.
Our algorithm is faster, as its time complexity is only quadratic.
Moreover, the algorithm uses an LCS-computation procedure as a component and any progresses in the LCS computation can improve the time complexities of the proposed method.

We also showed an irrecoverable flaw in~\cite{FFMR2010}, in which the algorithm of better than cubic time complexity was recently proposed, i.e., we proved this algorithm is supercubic in the worst case.

\subsection*{Acknowledgments}
The author thanks Szymon Grabowski for reading preliminary versions of the paper and suggesting improvements.

%\nocite*
%\bibliographystyle{plain}
%\bibliography{str-ic-lc-quadratic}

\end{document}